\documentclass[lettersize,journal]{IEEEtran}
\usepackage{amsmath,amsfonts}
\usepackage{algorithmic}
\usepackage{array}
\usepackage[caption=false,font=normalsize,labelfont=sf,textfont=sf]{subfig}
\usepackage{textcomp}
\usepackage{stfloats}
\usepackage{url}
\usepackage{verbatim}
\usepackage{graphicx}
\usepackage{amsmath}
\usepackage[ruled]{algorithm2e}
\usepackage{amsthm}
\usepackage{amssymb}

\newtheorem{lemma}{Lemma}

\usepackage{hyperref}
\newtheorem{proposition}{Proposition}
\usepackage{booktabs} 
\usepackage{multirow} 
\newtheorem{remark}{Remark}
\def\BibTeX{{\rm B\kern-.05em{\sc i\kern-.025em b}\kern-.08em
		T\kern-.1667em\lower.7ex\hbox{E}\kern-.125emX}}
\usepackage{balance}

\begin{document}
	\title{From Target Tracking to Targeting Track — Part III: Stochastic Process Modeling and Online Learning}
	\author{Tiancheng Li, \textit{IEEE Senior Member}, Jingyuan Wang, Guchong Li, Dengwei Gao %
		\thanks{Manuscript created Feb 2025; \\
			This work was supported in part by the National Natural Science Foundation of China under Grants 62422117 and 62201316 
			and in part by the Fundamental Research Funds for the Central Universities. 
			\\
			Tiancheng Li, Jingyuan Wang and Guchong Li are with the Key Laboratory of Information Fusion Technology (Ministry of Education), School of Automation, Northwestern Polytechnical University, Xi’an 710129, China, E-mail: t.c.li@nwpu.edu.cn, jy\_wang@mail.nwpu.edu.cn, guchong.li@nwpu.edu.cn. 
			Dengwei Gao is with No. 203 Research Institute of China Ordnance Industries, Xi'an 710065, China, E-mail: gaodengwei123@163.com. 
	}}
	
	\markboth{Journal of \LaTeX\ Class Files,~Vol.~, No.~, Feb~2025}%
	{How to Use the IEEEtran \LaTeX \ Templates}

	\maketitle
	
	\begin{abstract}
		This is the third part of a series of studies that model the target trajectory, which describes the target state evolution over continuous time, as a sample path of a stochastic process (SP). By adopting a deterministic-stochastic decomposition framework, we decompose the learning of the trajectory SP into two sequential stages: the first fits the deterministic trend of the trajectory using a curve function of time, while the second estimates the residual stochastic component through parametric learning of either a Gaussian process (GP) or Student's-$t$ process (StP). 
		This leads to a Markov-free data-driven tracking approach that produces the continuous-time trajectory with minimal prior knowledge of the target dynamics. Notably, our approach explicitly models both the temporal correlations of the state sequence and of measurement noises through the SP framework. 
		It does not only take advantage of the smooth trend of the target 
		but also makes use of the long-term temporal correlation of both the data noise and the model fitting error. 
		Simulations in four maneuvering target tracking scenarios have demonstrated its effectiveness and superiority in comparison with existing approaches. 
		
		
	\end{abstract}
	

	\IEEEpeerreviewmaketitle

	\section{Introduction}
	\label{sec:Introduction}
	
	\IEEEPARstart{T}{arget} tracking that involves the online estimation of the trajectory of a target has been a long-standing research topic and plays a significant role in aerospace, traffic, defense, robotics, etc. \cite{bar2004estimation} In essence, target tracking is more about estimating the continuous-time trajectory of the target rather than merely a finite number of point states. The continuous-time trajectory enables the acquisition of a point estimate of the state at any time in the trajectory period. However, the converse is not true. 
	Different from the classic \textit{target tracking} framework that estimates the point state $\mathbf{x}_k$ of the target on discrete measurement times $k \in \mathbb{N}^+$, 
	we solve the target tracking problem through modeling the target trajectory as a function of time (T-FoT). 
	That is, the evolution of the target state over time is modeled by T-FoT $f:\mathbb{R}^+ \rightarrow \mathcal{X}$, defined in spatio-temporal space, where $\mathcal{X}$ denotes the state space. 
	The target state at any time $t\in \mathbb{R}^+$ (not only at the measurement time instants) is given by  
	\begin{equation} \label{eq:T-FoT}
		{\mathbf{x}_t} = f(t).
	\end{equation}
	
	More attempts on continuous-time trajectory modeling can be found in the cutting-edge survey 
	\cite{Talbot2024continuoustimestate}. However, most of these approaches assume no or only short-term temporal correlation of the target state and of the measurement 
	and do not provide uncertainty about the trajectory estimate. 
	To take advantage of the latent long-term temporal correlation of the target state and of the measurement 
	in time series
	and to provide an assessment of the uncertainty associated with the estimate of the T-FoT, within our series of companion papers \cite{Li25TFoT-part1,Li25TFoT-part2} %
	including this one, we further model 
	the collection of target states as a (continuous-time) stochastic process (SP) $\mathcal{SP}_{x} \triangleq \{\mathbf{x}_t: t\in \mathbb{R}^+ \} $. 
	Therefore, any T-FoT is a sample path of this trajectory SP (TSP), that is, $f(t)\sim \mathcal{SP}_{x} $. 
	Based on the deterministic-stochastic decomposition (DSD) approach rooted in Wold and Cram\'er's decomposition theorem \cite{Wold1938,Cramer1961,Box1994}, the SP can be decomposed into a deterministic FoT and the residual SP (RSP), that is, 
	\begin{equation}
		f(t)=F(t;\mathbf{C})+\epsilon(t) , \label{eq:cramer}
	\end{equation}
	where the deterministic FoT $F(\cdot; \mathbf{C})$ is specified by parameters $\mathbf{C}$ and represents the trend of the SP $\mathcal{SP}_{x} $ and $\epsilon(\cdot)$ follows an RSP, denoted as $\mathcal{SP}_\epsilon$. 
	
	Note that $\epsilon(t)$ was interpreted as the fitting error of $F(t; \mathbf{C})$ to $f(t)$ in \cite{li2018joint,Li23Patent,li2023target} but no information is provided on the uncertainty of the estimate therein. Moreover, the temporal correlation between 
	the measurement noises in time series is ignored. 
	This paper, together with the companion paper \cite{Li25TFoT-part2} that focused on the trajectory trend FoT $F(\cdot; \mathbf{C})$, will overcome these issues. 
	This paper serves as the third part of this series and offers solutions for learning the RSP $\mathcal{SP}_\epsilon$ 
	for which two specific representative SPs are considered, respectively: the Gaussian process (GP) and Student's $t$ process (StP). 

	\begin{figure}[!htbp]
		\centering
		\includegraphics[width=0.47\textwidth]{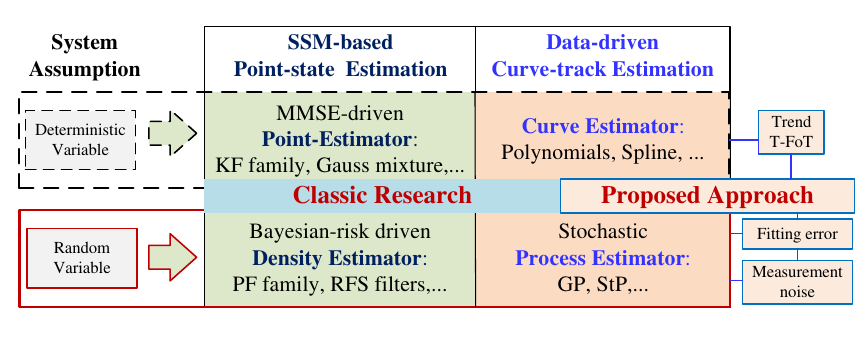}
		\caption{Comparison of our proposed approach with classic research}
		\label{fig:Innovative}
	\end{figure}

	\subsection{Relevant Work}
	\label{Literature Study}
	The classic target tracking approaches are based on the state space model (SSM) which provides a recursive filtering framework for estimation, based on which many minimum mean square error (MMSE)-based point estimators and Bayesian density estimators have been developed \cite{bar2004estimation,Sarkka13book}. In contrast to this, data-driven methods model the trajectory curve by a polynomial, splines, etc. since the pioneering work \cite{Rudd94,Wang94,Anderson-Sprecher96} 
	and, more relevantly, as an SP since batch trajectory modeling \cite{Tong12GP,Anderson15}. 
	The idea of SP learning is to use the prior distribution of a specified function and the training data to predict the function values of points not included in the training data \cite{williams2006gaussian}. This property has enabled its application in numerous time series analysis and prediction problems. 
	The two main types of estimators can be illustrated in Fig. \ref{fig:Innovative}. 
	
	\subsubsection{SP for tracking} The most widely used SP is the family of elliptical processes, 
	including the GP and StP, where any collection of function values has a desired elliptical distribution, with a covariance matrix constructed using a kernel. 
	Since the pioneering work \cite{ko2009gp,hartikainen2010kalman}, 
	the GP has been 
	incorporated into various calculation steps of SSM-based filters for tracking such as the Kalman filter \cite{williams2006gaussian,lyu2021adaptive} 
	and particle filter \cite{ko2009gp,sun2020gaussian,goodyer2023gapp,goodyer2023flexible,hu2022particle} or used individually in a recursive, online manner. In the latter, a representative method is the so-called GP motion tracker (GPMT) algorithm \cite{aftab2019gaussian,aftab2020learning}. 
	In GPMT, the target location is explicitly 
	modeled by a GP with a temporal argument, which can learn the target motion model online and switch flexibly. 
	It has been shown that different kernels and training data result in diverse tracking performance \cite{aftab2020impact} and factorisation methods have been proposed to expedite the GP learning \cite{lyu2022efficient}. Posterior Cram\'er-Rao bounds (CRBs) \cite{Zhao19CRLB-GP} and upper confidence bounds \cite{Liu23GPUpperbounds} are derived, respectively, for the performance of the GP-based SSMs. Moreover, the GP has also been extended to the multisensor system \cite{liu2022learning}. %
	More recently, a track segment association algorithm uses the local GP to predict and backtrack track segments \cite{guo2022joint}. Furthermore, the time-varying constant velocity model is incorporated into the GP learned online \cite{guo2023hybrid} to form a hybrid driving method to improve the prediction. There have also been numerous outstanding endeavors dedicated to modeling extended targets using the GP, e.g. \cite{Wahlstrom15ET-GP,Tang19GP-ET,Aftab19GP-ET,guo2020simultaneous,Qiu2024GPJPDA-ET}.  
	%
	While the GP itself is defined by the mean function and covariance kernel, the StP serves as an extension of the GP 
	introducing robustness to outliers and heavy-tailed measurement noise. 
	To gain robustness to outliers and heavy-tailed measurement noise, the StP has attracted much attention for non-Gaussian estimation and tracking, for example,  \cite{Agamennoni12Stu,Roth13Stu,solin2015state,Huang17Stu,Huang19StuKLD,Zhang21Stu}, just to name a few.

	In short, 
		existing SP-for-tracking approaches are mostly based on the SSM and 
		do not directly produce trajectory estimates. 
		The trajectory is given by the concatenation of optimal discrete-time point estimates, 
		which, however, does not necessarily yield the optimal trajectory estimate; see the illustration given in \cite{Chen2018}. In fact, there have been various optimization-based estimation methods such as the moving horizon estimator. \cite{zou2020moving,alessandri2011moving} and the batch trajectory GP \cite{Tong12GP,Anderson15,Rong2019shipGP,Chen20GPflight,Zhao21DeepGP,Nguyen2024GPTR}. 
		However, they are mostly still limited by the assumption of the Markov model and do not provide an online estimation of the continuous-time T-FoT. 
		The continuous-time T-FoT enables the acquisition of the state estimate immediately for any time in the fitting period. However, the converse is not true. 
		
		\subsubsection{Colored measurement noise}
		In the literature on target tracking, the statistical correction of measurement noise over time, a.k.a. colored noise \cite{Wu96colorN,Wang15color,Helgesen19colored,Yu20multiplicativeNoise}, is mostly overlooked. In fact, the impact of the colored noise on the CRB is non-negligible \cite{Lambert12CRBcolored,Yang25MUBcolorednoise}.A variety of approaches for handling colored noise have been devised which, however, are mostly constrained to the assumption of short-term temporal correlations. Specifically, the correlation often exists only between adjacent time instants. 
		It should be noted that measurement noises stem from three main aspects: the internal characteristics of the sensor such as thermal noise \cite{Helgesen19colored} and flicker noise, the external environmental conditions such as electromagnetic interference/glint \cite{Bilik10glint} and temperature fluctuations, and the properties of the measurement circuit such as power supply noise and redshift error \cite{Gao23redshift}. Contrary to the assumption of time-independence, they often exhibit significant cross-time correlations. In fact, some applications require the noise of a specific color \cite{Lee17color}. 
		Consequently, taking into account the temporal correlation of the measurement noise, not simply between adjacent time instants (as commonly assumed for colored noises in the literature of estimation), but in a long-term time range, is of great significance from both theoretical and practical perspectives. However, we are not aware of any previous tracker that models the colored noises as an SP. 
		

		\begin{figure}[htbp]
			\centering
			\includegraphics[width=0.46\textwidth]{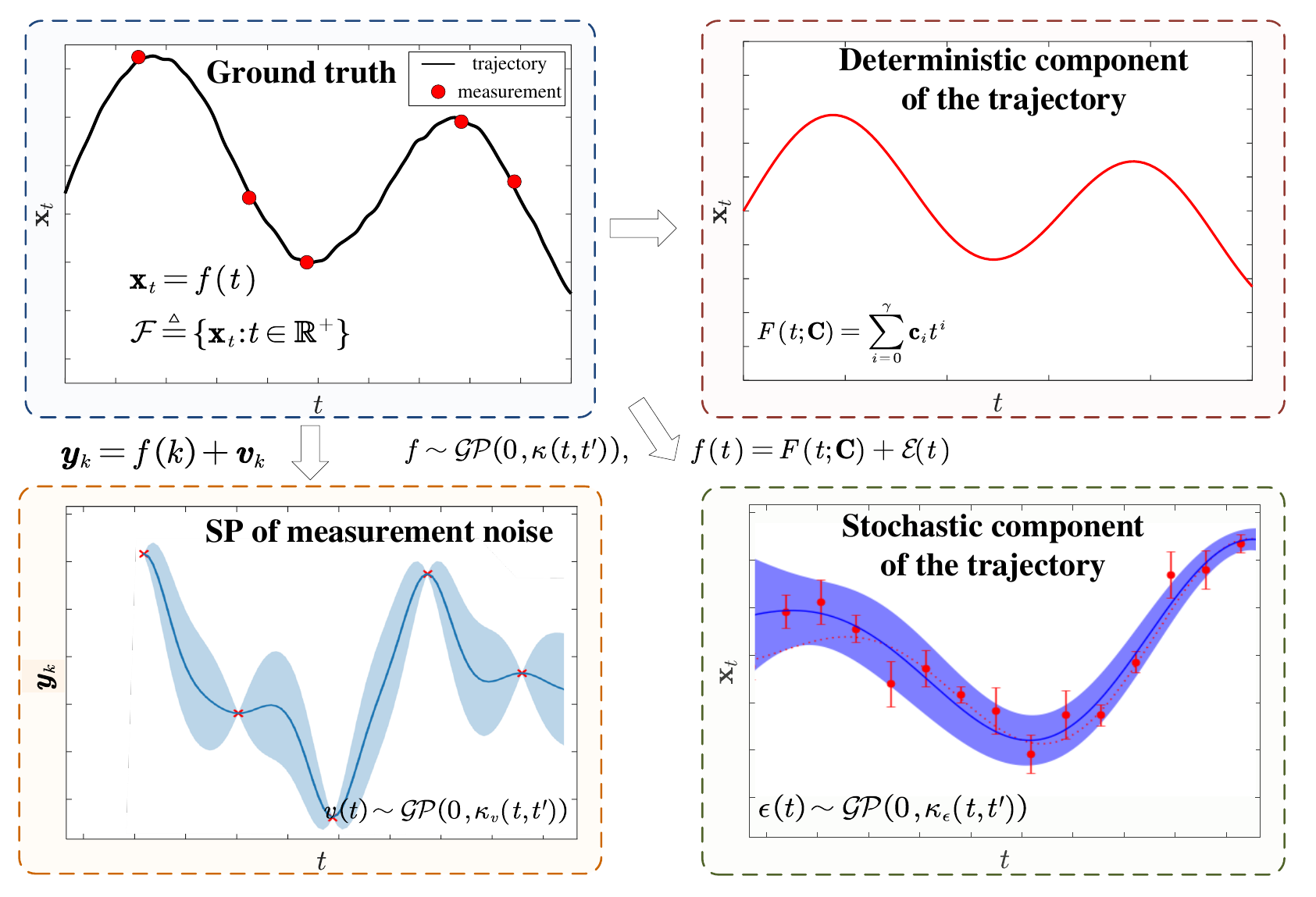}
			\caption{Key idea of the DSD approach to T-FoT estimation.}
			\label{fig:DSD_3_Comp}
		\end{figure}

		\subsection{Contribution and Organization of This Paper}
		The contribution and novelty of this paper include: 
		
		\begin{itemize}
			\item We redefine the classic target tracking problem as a T-FoT-oriented SP learning task based on DSD \eqref{eq:cramer}. Our approach is able to provide the trajectory estimate directly that comprises the trend T-FoT estimate $F(t;\mathbf{C})$ and the RSP $\mathcal{SP}_\epsilon$ approximated by the zero mean GP/StP. This leads to a computationally efficient hierarchical/two-stage estimation method in which we first estimate the deterministic trend component of the target trajectory and then refine our estimate by taking into account the temporal correlation and randomness via GP/StP to deal with the residual fitting error of the trajectory. 
			The T-FoT obtained in the first stage will be used as the mean function of the RSP in the second stage. Relevantly, we have documented a patent that used a recursive neural network to compensate for the T-FoT fitting error in \cite{Li23Patent} but no SP model was built. 
			\item We opt for using the GP/StP to approximate the RSP $\mathcal{SP}_\epsilon$ rather than attempting a direct approximation of the TSP $\mathcal{SP}_{x}$. This is simply because the TSP is typically complicated, non-stationary, and incorporates a significant deterministic component. However, with the use of DSD, as elaborated in \cite{Cramer1961,Box1994}, it is reasonable and justifiable to assume that the measurement noise, the RSP adheres to the properties of stationarity (at least approximately) and has a zero mean, separately. 
			In contrast, the assumption that the TSP is entirely stationary (and even has a zero mean) is arguably overly restrictive and unrealistic, which fails to account for the compoundness nature of the trajectories involved in practice. 
			\item We model the temporal correlation between the measurement noises using an SP too. The resulting SP learning approach is Markov/SSM-free and advantageous in representing the long-term correlation between the data. 
			Experimental results show that our proposed method outperforms traditional GP/StP-for-tracking approaches that use no DSD and overlook the temporal correlation between the measurements, especially for maneuvering target tracking in complex scenarios. 
		\end{itemize}

		The DSD approach that has been applied in our approach can be illustrated in Fig. \ref{fig:DSD_3_Comp}. As shown, our approach deals with the deterministic trend, the T-FoT fitting error, and the measurement noise, separately. The remainder of this paper is organized as follows. Section \ref{sec:background} provides preliminaries on GP and StP. Section \ref{sec:TSP Tracking} describes the proposed tracking method in detail. Section \ref{sec:Simulations} presents the performance evaluation and experimental results. The paper is concluded in Section \ref{sec:Conclusion}. The key 
		notations used in this paper are given in Table \ref{tbl:symbols}.

		\begin{table}[htbp]
			\centering
			\caption{List of key notations}
			\begin{tabular}{ll}
				\toprule
				Notation & Interpretation \\
				\midrule
				$t$ & $t\in \mathbb{R}^+$, continuous time (positive real number) \\
				$\mathbf{t}_k$ & $\mathbf{t}_k \triangleq \left[ k',k'+1\cdots ,k\right] ^{\top}$, time window \\
				$\mathbf{x}_t, \mathcal{SP}_{x}$ &target state at time $t$ and the corresponding SP \\
				$k$ & $k\in \mathbb{N}^+$, discrete time (positive integer) \\
				$\mathbf{y}_k$ & measurement at time $k$ \\
				${\mathbf{H}}$ & linear measurement matrix at time $k$ \\
				$\mathbf{v}_k$ & measurement noise at time $k$ \\
				$\sigma^{2}\mathbf{I}$ & variance of the measurement noise \\
				$F\left ( t;\mathbf{C} \right )$ & polynomial T-FoT specified by parameters $\mathbf{C}$   \\
				$d$ & maximum length of the sliding time window \\
				$\mathcal{N}(\mathbf{m},\mathbf{K})$ & normal distribution with mean $\mathbf{m}$, covariance $\mathbf{K}$ \\ 
				$\mathcal{GP}(m(t),\kappa(t, t'))$ &  GP with mean $m(t)$ and covariance $\kappa(t, t')$\\
				$\mathcal{T}(\mathbf{m},\mathbf{K},\nu)$ & Student's $t$ distribution with mean $\mathbf{m}$, \\
				& covariance $\mathbf{K}$ and degree of freedom $\nu$ \\ 
				$\mathcal{TP}(m(t),\kappa(t, t'),\nu)$ &  StP with mean $m(t)$, covariance $\kappa(t, t')$ \\
				& and degree of freedom $\nu$\\
				$g(k),\mathcal{SP}_{g} $ & pseudo measurement at time $k$ and its SP\\
				$e(k),\mathcal{SP}_{e} $ & measurement fitting error at time $k$ and its SP\\ 
				$\epsilon(t),\mathcal{SP}_{\epsilon} $ & T-FoT fitting residual at time $t$ and its SP \\ 
				$v(k),\mathcal{SP}_{v}$ & measurement noise at time $k$ and its SP \\
				$\mathbb{C}\left( a , b\right)$ & covariance between $a$ and $b$ \\
				$\varepsilon_{m}$, $\zeta$ & the variance and length scale of the RBF kernel \\
				\bottomrule
			\end{tabular}
			\label{tbl:symbols}
		\end{table}

		\section{Preliminaries}
		\label{sec:background}
		
		In this paper, the following assumptions that have been widely used in the literature, e.g. \cite{aftab2019gaussian,aftab2020learning}, are made. 
		\begin{enumerate}
			\item[A1] Cross-coordinates coupling is weak enough to be ignored. This enables us to carry out inference in each coordinate individually for computation efficiency. 
			\item[A2] The state to infer is fully observable. More specifically, this paper focuses on the linear measurement of the target position, and the T-FoT is for localization. 
		\end{enumerate}

		\subsection{GP Regression}
		
		We briefly review the formulation of univariate GP regression, 
		which is applicable to multivariate inputs and outputs in a similar way. 
		Let $\mathcal{T}$ be the input space. For any finite input set $\boldsymbol{t} \triangleq \left\{ t_1,t_2,\cdots ,t_n \right\} \subseteq \mathcal{T}$ , if the set of random variables $\mathbf{f} \triangleq [f(t_1),f(t_2),\cdots,f(t_n)]^{\top}$  follows a joint normal distribution, that is, 
		\begin{equation} 
			\mathbf{f} \sim \mathcal{N}(\mathbf{m}(\boldsymbol{t}),\mathbf{K}(\boldsymbol{t},\boldsymbol{t})), 
		\end{equation} 
		where $\mathbf{m}(\boldsymbol{t})=[m(t_1),m(t_2),\cdots,m(t_n)]^{\top}$ is the mean vector, and the elements of the covariance matrix $\mathbf{K}(\boldsymbol{t},\boldsymbol{t})$ are $K_{ij}=\kappa(t_i,t_j)$, for $i,j = 1,2,\cdots,n$. 
		Then, the function $f:\mathcal{T}\to\mathbb{R}$ follows a G, denoted as 
		\begin{equation} 
			f(t)\sim \mathcal{GP}(m(t),\kappa(t,t')) , \label{eq:GP_mapping} 
		\end{equation}
		where the mean function $m(\cdot):\mathcal{T}\to\mathbb{R}$ is defined as $m(t)=\mathbb{E}[f(t)]$, 
		and the covariance kernel function $\kappa(\cdot,\cdot):\mathcal{T}\times\mathcal{T}\to\mathbb{R}$ is defined as $\kappa(t,t') = \mathbb{C}[f(t),f(t')]$. 
		
		
		
		
		Consider the observed data $\mathbf{y}= \mathbf{H} \mathbf{f} +\mathbf{v} $, where $\mathbf{H}$ is a linear matrix that has full row rank, $\mathbf{v}$ is a noise vector, and $\mathbf{v }\sim \mathcal{N}\left( 0,\sigma ^{2}\mathbf{I} \right)$.
		According to the property of the GP, given a collect of $n$ data $\boldsymbol{\mathbf{y}}=[{y}_1,{y}_2,\cdots,{y}_n]^{\top}$, the predicted value $f_{\ast}$ at a new input point $t_{\ast}$ is also normally distributed: 
		\begin{equation} 
			p(f_{\ast}|t_{\ast},\boldsymbol{t},\boldsymbol{\mathbf{y}}) = \mathcal{N}(\mu_{\ast},\Sigma_{\ast}), 
		\end{equation} 
		where the predicted mean $\mu_{\ast}$ and predicted covariance $\Sigma_{\ast}$ are 
		\begin{equation} 
			\mu_{\ast} = m({t}_{\ast}) + \mathbf{k}_*^{\top}\mathbf{J}_{\boldsymbol{t}} \big(\boldsymbol{\mathbf{y}}-\mathbf{m}(\boldsymbol{t})\big), 
			\label{eq:GPR_mean} 
		\end{equation} 
		\begin{equation} 
			\Sigma _{\ast} = \kappa\left( t_{\ast},t_{\ast} \right) -\lVert \mathbf{k}_* \rVert _{\mathbf{J}_{\boldsymbol{t}}},
			\label{eq:GPR_covariance} 
		\end{equation} 
		where $\mathbf{k}_*=[\kappa(t_{\ast},t_1),\kappa(t_{\ast},t_2),\cdots,\kappa(t_{\ast},t_n)]^{\top}$ denotes the vector of covariances between the test point $t_{\ast}$ and the training points and $\lVert \mathbf{k} \rVert _{\Sigma}\triangleq \mathbf{k}^{\top}\Sigma ^{-1}\mathbf{k}$. There are two different ways to set $\mathbf{J}_{\boldsymbol{t}}$ according to the real need: one assumes noise-free data and uses  $\mathbf{J}_{\boldsymbol{t}}=\mathbf{K}\left( \boldsymbol{t,t} \right)^{-1}$ and the other assumes noisy data and uses $\mathbf{J}_{\boldsymbol{t}}=\left[ \mathbf{K}\left( \boldsymbol{t,t} \right) +\sigma ^{2}\mathbf{I} \right] ^{-1}$.

		\subsection{StP Regression}
		
		For any finite input set $\boldsymbol{t} \triangleq \left\{ t_1,t_2,\cdots ,t_n \right\} \subseteq \mathcal{T}$ , if the set of random variables $\mathbf{f} \triangleq [f(t_1),f(t_2),\cdots,f(t_n)]^{\top}$ follows a joint Student's-$t$ distribution, that is, 
		\begin{equation} 
			\mathbf{f}\sim \mathcal{T}\left( \mathbf{m}\left( \boldsymbol{t} \right) ,\frac{\nu-2}{\nu}\mathbf{K}\left( \boldsymbol{t,t} \right) ,\nu \right) ,
		\end{equation} 
		where $\mathcal{T}\left( \cdot \right) $ represents the Student's-$t$ distribution. 
		Then, the function $f:\mathcal{T}\to\mathbb{R}$ follows an StP, denoted as
		\begin{equation} 
			f(t)\sim \mathcal{TP}(m(t),\kappa(t,t'),\nu),
			\label{eq:StP_mapping} 
		\end{equation}
		where $\nu >2$ is the degrees of freedom (DoF), controlling the tail behavior of the distribution. In particular, under the same kernel, as $\nu \to \infty$, the StP converges to a GP \cite{shah2014student}.

		
		Consider the observed data $\mathbf{y}= \mathbf{H} \mathbf{f} +\mathbf{v} $, where $\mathbf{v }\sim \mathcal{T}\left( 0,\sigma ^{2}\mathbf{I},\nu \right)$. 
		According to the property of the StP, 
		the predicted value $f_{\ast}$ at a new input point $t_{\ast}$ is $t$ distributed: 
		\begin{equation} 
			p\left( f_{\ast}|t_{\ast},\boldsymbol{t},\mathbf{y} \right) =\mathcal{T}\left( \mu _{\ast},\Sigma_{\ast},\nu_{\ast} \right) ,
		\end{equation} 
		where the predicted mean $\mu_{\ast}$ is the same calculated as in \eqref{eq:GPR_mean} and the predicted covariance $\Sigma_{\ast}$ and DoF $\nu _{\ast}$ are given, respectively, as follows 
		\begin{equation} 
			\Sigma _{\ast}=\frac{\nu +\lVert \mathbf{y} \rVert _{\mathbf{J}_{\boldsymbol{t}}}-2}{\mathbf{\nu }+\left| \boldsymbol{t} \right|-2}\left( \kappa\left( t_{\ast},t_{\ast} \right) -\lVert \mathbf{k}_* \rVert _{\mathbf{J}_{\boldsymbol{t}}} \right) ,
			\label{eq:StPR_covariance} 
		\end{equation} 
		\begin{equation} 
			\nu _{\ast}=\nu +\left| \boldsymbol{t} \right| , 
			\label{eq:StPR_dof} 
		\end{equation}
		with the same definition of $\mathbf{k}_*,\mathbf{J}_{\boldsymbol{t}}$ as in \eqref{eq:GPR_mean}.

		\section{T-FoT GP Modeling and Learning}
		\label{sec:TSP Tracking}
		In this section, the data model is first given before we present our proposed trajectory learning approach.
		
		\subsection{Data Model}
		Following the T-FoT model \eqref{eq:T-FoT} and TSP \eqref{eq:cramer}, we consider the linear measurement function with additive noise, that is, 
		\begin{align}   
			\mathbf{y}_k
			&= {\mathbf{H}} \mathbf{x}_k  + \mathbf{v}_k , \nonumber \\
			& = {\mathbf{H}} F(k;\mathbf{C})+ {\mathbf{H}} \epsilon(k) + \mathbf{v}_k , \label{eq:Meas-h-DSD}
		\end{align}
		where ${\mathbf{H}}$ denotes the linear measurement matrix that has full row rank, 
		$\mathbf{y}_k \in \mathcal{Y}$ %
		and $\mathbf{v}_k $ 
		denote the colored measurement and the zero mean noise at time $k\in \mathbb{N}^+$, respectively, and $\mathcal{Y}$ denotes the measurement space. In this paper, we employ the GP and StP to model the temporal correction of the measurement noises, in Sec. \ref{sec:T-FoT-GP} and \ref{sec:T-FoT-StP-algorithm}, respectively. 
		
		Based on the decomposition of the measurement \eqref{eq:Meas-h-DSD}, the following definitions are made
		\begin{align}
			g(k) & \triangleq {\mathbf{H}} \epsilon(k) , \label{eq:def_g} \\
			e(k) & \triangleq \mathbf{y}_k-{\mathbf{H}}F(k;\mathbf{C}) , \label{eq:def_3} 
		\end{align}
		where $g(k)$ is called the pseudo measurement hereafter, $e(k)$ denotes the measurement fitting error at time $k$ generated in fitting the trend of the T-FoT \cite{li2018joint,li2023target,Li25TFoT-part2} and it is easy to get
		\begin{equation} \label{eq:g=e+v}
			g(k) = e(k) -\mathbf{v}_k .
		\end{equation}
		
		Since the linear transformation of an SP results in another SP, 
		the pseudo measurement $g(k)$ also follows an SP, denoted as $\mathcal{SP}_g$. In what follows, we will elaborate on the procedures for learning TSP $\mathcal{SP}_{x}$ in specific forms of GP and StP, respectively. 
		To prevent the need for an overly high-order T-FoT, 
		a time window $\mathbf{t}_k \triangleq \left[ k',k'+1,\cdots ,k\right] ^{\top}$ of an appropriate length that slides with time is needed to control the data size for the fitting of T-FoT and the learning of TSP, where $k'= \max(k-d, 1)$, $d$ is the maximum length of the sliding time window. This complies with the default setup of the T-FoT fitting \cite{li2018joint,li2023target,Li25TFoT-part2}, which differs from the batch GP trajectory approaches \cite{Tong12GP,Anderson15,Rong2019shipGP,Chen20GPflight,Zhao21DeepGP,Nguyen2024GPTR}. Unless explicitly stated otherwise, the measurement fitting errors $\mathbf{e}({\mathbf{t}_k})$ received in the sliding time window $\mathbf{t}_k$ are the training data to learn the RSP $\mathcal{SP}_\epsilon$ and then the TSP $\mathcal{SP}_{x}$. 
		
		\begin{algorithm}[t]
			\caption{Flowchart of DSD-based TSP Learning}\label{AlgorithmTable}
			\begin{algorithmic}
				\setlength{\lineskip}{3pt}
				\setlength{\lineskiplimit}{3pt}
				\STATE \textbf{Input} $\mathbf{y}_{\mathbf{t}_k}, {\mathbf{H}}, \mathbf{\sigma}^2, d$ 
				\STATE \textbf{Output} TSP $\mathcal{SP}_{x} $
				\STATE \hspace{0.5cm}\textbf{Step 1.} calculate $F(t;\mathbf{C}_k)$ via \eqref{eq:C_learning}. 
				\STATE \hspace{0.5cm}\textbf{Step 2.} {generate pseudo data} $e(k)=\mathbf{y}_k-{\mathbf{H}}F(k;\mathbf{C}_k)$. 
				\STATE \hspace{0.5cm}\textbf{Step 3.}  train the measurement fitting error SP $\mathcal{S} \mathcal{P}_e $ using the pseudo data, that is, $e\left( k \right) \sim \mathcal{S} \mathcal{P}_e$. 
				\STATE \hspace{0.5cm}\textbf{Step 4.} determine $\mathcal{S} \mathcal{P}_g$ based on $\mathcal{S} \mathcal{P}_e$ via \eqref{eq:g=e+v}
				\STATE \hspace{0.5cm}\textbf{Step 5.} determine $\mathcal{S} \mathcal{P}_{\epsilon}$ based on $\mathcal{S} \mathcal{P}_g$ via \eqref{eq:def_g}.
				\STATE \hspace{0.5cm}\textbf{Step 6.} determine TSP $\mathcal{SP}_{x} $ based on $F(t;\mathbf{C}_k)$ and $\mathcal{S} \mathcal{P}_{\epsilon}$ via \eqref{eq:cramer}
			\end{algorithmic}
		\end{algorithm}

		

		\subsection{T-FoT GP Modeling and Decomposition}
		\label{sec:T-FoT-GP}
		The following assumption is made to account for the temporal correlation between the measurement noises. 
		\begin{enumerate}
			\item[A3] The collection of the measurement noises corresponding to any time follows a multivariate zero mean Gaussian distribution, that is, the measurement noises follow a discrete time zero mean GP 
			$$ \mathcal{SP}_v \triangleq \{\mathbf{v}_k; k\in \mathbb{N}^+\} \sim \mathcal{GP}(\mathbf{0},\kappa_v(k,k')),$$  
			where $\boldsymbol{0}$ denotes a zero vector of appropriate dimension.
		\end{enumerate}
		
		\begin{remark}
			Independent and identically distributed (IID) Gaussian noise can be viewed as a specific case of the above GP assumption with $\kappa_v(k,k') = \mathbf{0}$ if $k\neq k'$ and $ \kappa_v(k,k) = \kappa_v(k',k'), \forall k, k' \in \mathbb{N}^+$.  
		\end{remark}
		
		Our first main result is on modeling both the target states and measurement noises by continuous- and discrete-time GPs, respectively. To this end, several significant statistical properties of the GP shall be used. 
		
		\begin{lemma} \label{lemma_linear_Function_GP}
			The linear transformation of a GP ${\epsilon}(k) \sim \mathcal{GP}(m_{\epsilon}(k),\kappa_{\epsilon}(k, k'))$ is still a GP. Specifically, $g(k) \triangleq {\mathbf{H}} \epsilon(k) \sim \mathcal{GP}({\mathbf{H}} m_{\epsilon}(k),{\mathbf{H}} \kappa_{\epsilon}(t,t'){\mathbf{H}}^{\top})$.
		\end{lemma}
		\begin{proof}
			This lemma is a well-known property of GP \cite{williams2006gaussian}. In detail, the calculation of the mean and covariance functions is given in Appendix \ref{appendx:prof_lemma_linear_Function_GP}. 
		\end{proof}
		
		\begin{lemma} \label{lemma_GP-mt}
			Assuming that the TSP $\mathcal{SP}_{x}$ is a G, that is,  $f(t)\sim \mathcal{GP}(m(t),\kappa(t,t'))$ and further assuming that the fitted T-FoT $F(t;\mathbf{C})$ is chosen as the mean function $m(t)$ of the trajectory GP, 
			the RSP $\mathcal{SP}_\epsilon \triangleq\{ \epsilon(t):=f(t)-m(t); t\in \mathbb{R}^+\} $ is a zero-mean GP with the same covariance function $\kappa(t,t')$. 
		\end{lemma}
		\begin{proof}
			This lemma is also a well-known property of the GP \cite{williams2006gaussian}. In detail, the calculation of the mean and covariance functions is given in Appendix \ref{appendx:prof_lemma_GP-mt}. 
		\end{proof}
		
		\begin{lemma} \label{lemma_GP_plus_GP}
			The sum $ e\left( k \right) ={g}\left( k \right) +v\left( k \right) $ of two independent GPs ${g}(k) \sim \mathcal{GP}(m_{g}(k),\kappa_{g}(k, k'))$ and $ v(k)\sim \mathcal{GP}(m_{v}(k),\kappa_{v}(k, k'))$ is also a GP $e(k)\sim\mathcal{GP}(m_{g}(k)+m_{v}(k),\kappa_{g}(k, k')+\kappa_{v}(k, k'))$.
		\end{lemma}
		\begin{proof} 
			This lemma is also a well-known property of the GP. The calculation of the mean and covariance functions is given in Appendix \ref{appendx:proof_lemma_GP_plus_GP}. 
		\end{proof}

		
		
		Based on the above theoretical results, 
		the learning of the TSP can be decomposed into the following two stages. 

		\subsubsection{First stage (corresponding to Step 1 of Algorithm \ref{AlgorithmTable})}
		We first learn the trend, deterministic part of the TSP $\mathcal{SP}_{x}$ from the measurement in the time window $\mathbf{t}_k$ by assuming a polynomial T-FoT $F\left( t;\mathbf{C}_k \right)$ of $\gamma$-order 
		as follows
		\begin{equation}\label{eq:polynomial}
			F\left ( t;\mathbf{C}_k \right )= \sum_{i=0}^{\gamma }\mathbf{c}_{i}t^i,
		\end{equation}
		where $\mathbf{C}_k=\left\{ \mathbf{c}_i \right\} _{i=0,1,\cdots ,\gamma}$ represents the polynomial trajectory coefficients, $\mathbf{c}_i= \big[c_i^{(1)}, c_i^{(2)}, \cdots, c_i^{(r)}  \big]$, $r$ indicates the dimension in the state space $\mathcal{X}$, 
		the order $\gamma$ should be optimized online as addressed in \cite{Li25TFoT-part2} if it is not available beforehand.  
		
		The trend T-FoT parameters $\mathbf{C}_k$ are estimated by solving the following optimization at the latest time $k$,
		\begin{equation} \label{eq:C_learning}
			\mathbf{C}_k=\underset{\mathbf{C}}{\text{arg}\min}\sum_{t \in \mathbf{t}_k} {\lVert \mathbf{y}_t-  F\left( t; \mathbf{C}\right) \rVert} ,
		\end{equation}
		where $\lVert a-b \rVert $ represents a distance between $a, b$; when the $\ell_2$ norm or Mahalanobis distance is agreed upon, that is, least squares (LS) or weighted LS \cite{li2018joint,li2023target}, which actually admits exact optimal solution in the case of linear measurement. 
		

		
		
		\subsubsection{Second stage} 
		Based on Lemma \ref{lemma_GP-mt}, we set the mean function $m(t)$ of the TSP $\mathcal{SP}_{x}$ as the above learned T-FoT $F(t;\mathbf{C}_k)$, resulting in the RSP $\mathcal{SP}_\epsilon$ being a zero mean GP with the same covariance kernel as that of the TSP $\mathcal{SP}_{x} $; see Step 5 of Algorithm \ref{AlgorithmTable}
		\begin{equation}
			\epsilon\left( t \right) \sim \mathcal{G}\mathcal{P}\left( \boldsymbol{0},\kappa_{\epsilon}\left( t, t' \right) \right) ,
			\label{eq:epsilon_t_modeling}
		\end{equation}
		where $\kappa_{\epsilon}\left( t, t' \right) =\kappa \left( t, t' \right)$.
		
		Further using Lemma \ref{lemma_linear_Function_GP}, we get, corresponding to Step 4 of Algorithm \ref{AlgorithmTable} 
		\begin{equation}\label{eq:g_t_modeling}
			g\left( k \right) \sim \mathcal{G}\mathcal{P}\left( \boldsymbol{0}, \kappa_g \left( k, k' \right) \right) ,
		\end{equation}
		where $ \kappa_g \left( k, k' \right)  \triangleq {\mathbf{H}} \kappa_{\epsilon} \left( k, k' \right){\mathbf{H}}^{\top} ={\mathbf{H}} \kappa \left( k, k' \right){\mathbf{H}}^{\top}  $

		Therefore, we only need to learn the covariance kernel of the above RSP $\mathcal{SP}_\epsilon$ based on the measurement fitting errors $ \mathbf{e}(\mathbf{t})$ (corresponding to Step 2 of Algorithm \ref{AlgorithmTable}).
		Then, given Assumption A3, one can deduce from Lemma \ref{lemma_GP_plus_GP} that, corresponding to Step 3 of Algorithm \ref{AlgorithmTable}
		\begin{equation}
			e\left( k \right) \sim \mathcal{G}\mathcal{P}\left( \boldsymbol{0}, \kappa_e \left( k, k' \right) \right),
			\label{eq:e_t_modeling}
		\end{equation}
		where $\kappa_e \left( k, k' \right) =\kappa_g \left( k, k' \right) + \kappa_v \left( k, k' \right)$. 
		%
		
		Given that the measurement noise covariance $\kappa_v \left( k, k' \right)$ is assumed known, 
		and $\kappa_e \left( k, k' \right)$ can be learned from the data as shown next, then the covariance kernel of the TSP $\mathcal{SP}_{x}$ can be calculated by, corresponding to Step 6 of Algorithm \ref{AlgorithmTable}
		\begin{equation}
			\kappa \left( k, k' \right) = \mathbf{A}_k \big(\kappa_e \left( k, k' \right)-\kappa_v \left( k, k' \right)\big) \mathbf{B}_k ,
			\label{eq:e_kernel}
		\end{equation}
		where $\mathbf{A}_k $ is the left inverse of ${\mathbf{H}}$ and $\mathbf{B}_k $ is the right inverse of ${\mathbf{H}}^{\top}$, that is $ \mathbf{A}_k {\mathbf{H}}= \mathbf{I}, {\mathbf{H}}^{\top}\mathbf{B}_k= \mathbf{I}$.

		\subsection{Hyperparameter Learning}
		\label{RGPlearning}
		To illustrate the learning of the hyperparameters of the measurement fit error SP $\mathcal{SP}_{e}$ from the data (corresponding to Step 3 of Algorithm 1), the kernel of radial basis function (RBF) is considered in our approach. 
		\begin{equation}
			\kappa \left( t, t^{\prime} \right) = \varepsilon_{m}^{2} \exp \left( -\frac{\| t-t^{\prime} \|^{2}}{2 \zeta^{2}} \right) ,
			\label{eq:RBF kernel}
		\end{equation}
		where the hyperparameters $\varepsilon_{m}^2$ and $\zeta$ represent the variance and length scale of the RBF kernel, respectively.

		The goal of learning the RSP $\mathcal{SP}_\epsilon$ online in the second stage of our proposed T-FoT-GP approach is to learn hyperparameters $\boldsymbol{\eta }_k \triangleq \{\varepsilon_{m}^2, \zeta\}$ simultaneously with the estimation of the values of function $e(\cdot)$. For this purpose, we employ the maximum likelihood approach referred to as recursive GP* (RGP*) \cite{huber2014recursive}; see also \cite{aftab2020learning} and theoretical study on the GP parameter estimation \cite{Karvonen20MLforGP}. 
		Finally, the T-FoT and its covariance at time $k$ are given as follows:
		\begin{align}
			\hat{f}_k\left( t \right)&=\hat{\mu}_{k}^{\bar{\epsilon}}+F\left( t;\mathbf{C}_k \right) \nonumber  \\
			&= \mathbf{A}_k \hat{\mu}_{k}^{\bar{e}} +F\left( t;\mathbf{C}_k \right)  \nonumber \\
			& = \mathbf{A}_k  \big( \mathbf{k}_k^{\top}\mathbf{J}_{\mathbf{t}_k} \mathbf{y}_k \big) +F\left( t;\mathbf{C}_k \right)  ,
			\label{eq:estimated_f_mean_GP}
		\end{align}
		\begin{equation}
			\mathbb{C}\left[ \hat{f}_k\left( t\right) \right] = \kappa\left( k,k \right) -\lVert \mathbf{k}_k \rVert _{\mathbf{J}_{\mathbf{t}_k}} , 
			\label{eq:estimated_f_cov_GP}
		\end{equation}
		where $\hat{\mu}_{k}^{\bar{\epsilon}}$ and $\hat{\mu}_{k}^{\bar{e}}$ denote the estimated mean of the functions $\epsilon\left(\cdot\right)$ and $e\left(\cdot\right)$, respectively, evaluated at measurement time $k$, 
		\eqref{eq:GPR_mean} is used in \eqref{eq:estimated_f_mean_GP} with $\mathbf{m}(\boldsymbol{t})=\mathbf{0}$ and \eqref{eq:GPR_covariance} is used in \eqref{eq:estimated_f_cov_GP}, 
		$\mathbf{k}_k$ is calculated by \eqref{eq:e_kernel}. 
		
		\begin{remark}
			Since we model the measurement noise as SP $\mathcal{SP}_v$ and subtract it from the fitting error SP $\mathcal{SP}_e$, it is the noise-free GP regression that should be used in calculating \eqref{eq:GPR_mean} to estimate the parameters of $\mathcal{SP}_e$. 
		\end{remark}

		\section{T-FoT StP Modeling and Learning}
		\label{sec:T-FoT-StP-algorithm}
		In place of Assumption A3, the following assumption may be made to account for the temporal correlation between measurement noises with heavier tails. 
		
		\begin{enumerate}
			\item[A4] The collection of the measurement noises corresponding to any time follows a multivariate zero mean Student's $t$ distribution, that is, the measurement noises follow a discrete time zero mean StP, where $\nu_v>2$, 
			$$ \mathcal{SP}_v \triangleq  \{\mathbf{v}_k; k\in \mathbb{N}^+\} \sim \mathcal{TP}(\mathbf{0},\kappa_v(k,k'),\nu_v).$$
		\end{enumerate} 
		
		Our second main result is on modeling both the target states and measurement noises by continuous- and discrete-time StPs, respectively. 
		However, unlike the case of GP, the StP is not amenable to translation and decomposition while maintaining the joint Student's $t$ distribution. Specifically, we have the following three known results. 
		
		\begin{lemma} \label{lemma_linear_Function_StP}
			Given an StP ${\epsilon}(k) \sim \mathcal{TP}(m_{\epsilon}(k),\kappa_{\epsilon}(k, k'),\nu_{\epsilon})$, its linear transformation $g(k) \triangleq {\mathbf{H}} \epsilon(k)$ is generally no more an StP but an SP with mean ${\mathbf{H}} m_{\epsilon}(k)$, covariance $\frac{\nu_{\epsilon}}{\nu_{\epsilon}-2} {\mathbf{H}} \kappa_{\epsilon}(k, k'){\mathbf{H}}^{\top}$ and DoF $\nu_{\epsilon}$. 
		\end{lemma}
		\begin{proof}
			It is easy to see that the mean, covariance and DoF of ${\epsilon}(k)$ are $m_{\epsilon}(k),\frac{\nu_{\epsilon}}{\nu_{\epsilon}-2} \kappa_{\epsilon}(k, k'), \nu_{\epsilon}$, respectively. Based on this, the mean and covariance of the transformed StP can be calculated similarly as shown in \eqref{eq:GP_linearTrasf_mean} and \eqref{eq:GP_linearTrasf_cova}. 
			Moreover, the linear transformation does not change the DoF of the StP. 
		\end{proof}
		
		\begin{lemma} \label{lemma_StP}
			Assuming that the TSP $\mathcal{SP}_{x}$ is an StP, that is,  $f(t)\sim \mathcal{TP}(m(t),\kappa(t,t'),\nu )$ and further assuming that the fitted T-FoT $F(t;\mathbf{C}_k)$ is chosen as the mean function $m(t)$ of the learning TSP $\mathcal{SP}_{x}$, the RSP $\mathcal{SP}_\epsilon$ is a zero mean SP but not an StP. 
		\end{lemma}
		\begin{proof} The key conclusion that the RSP is a shifted version of the original Student's \(t\) distribution but not an StP can be found in  \cite[Ch.3]{Degroot2013probability}. The resulted mean function can be the same calculated as in \eqref{eq:mean_f-m=0}. 
			Here, we only calculate its mean function 
		\end{proof}
		
		\begin{lemma} \label{lemma_StP_plus_StP}
			The sum $ e\left( t \right) =g \left( t \right) +v\left( t \right) $ of two independent StPs ${g}(k) \sim \mathcal{TP}(m_{g}(k),\kappa_{g}(k, k'),\nu_{g})$ and $ v(k)\sim \mathcal{TP}(m_{v}(k),\kappa_{v}(k, k'),\nu_{v}(k))$ is an SP with mean $m\left( t \right) = m_{g}\left( t \right) +m_{v}\left( t \right) $ but not an StP. 
		\end{lemma}
		\begin{proof} 
			The key conclusion that the sum of two independent StP is no more an StP can be found in \cite[Ch.3]{Degroot2013probability}. The mean function can be calculated the same as in \eqref{eq:e-mean_mg+mv} and the calculation of the covariance is given in Appendix \ref{appendx:proof_lemma_StP_plus_StP}.  
		\end{proof}
		
		
		To facilitate translation and decomposition of StP, we propose the following three moment matching methods for SP approximation while the DoF that does not have a specific form can be determined in a straightforward and ad-hoc way. 
		
		
		
		\begin{proposition}
			For the case described in Lemma \ref{lemma_linear_Function_StP}, the approximate StP ${\hat{g}}(k) \sim \mathcal{TP}(m_{g}(k),\kappa_{g}(k, k'),\nu_{g})$ that has the same moments as the RSP $\mathcal{SP}_\epsilon$, ${\mathbf{H}} \epsilon(k)$, is given as follows 
			\begin{align}
				m_{g}(k) & = {\mathbf{H}} m_{\epsilon}(k) ,\nonumber \\
				\kappa_{g}(k, k') & = {\mathbf{H}} \kappa_{\epsilon}(k, k'){\mathbf{H}}^{\top}, \nonumber \\
				\nu_{g} &= \nu_{\epsilon},
			\end{align}
		\end{proposition} 
		\begin{proof}
			The result is straightforward based on Lemma \ref{lemma_linear_Function_StP}.
		\end{proof}

		\begin{proposition}
			For the case described in Lemma \ref{lemma_StP}, the approximate StP that has the same moments as the RSP $\mathcal{SP}_\epsilon \triangleq\{ \epsilon(t):=f(t)-m(t); t\in \mathbb{R}^+\}$ is given as follows 
			\begin{equation}
				\hat{\epsilon}(t) \sim  \mathcal{TP}(\mathbf{0},\kappa_{\epsilon}\left( t, t' \right) ,\nu_{\epsilon} ),
			\end{equation}
			where $\kappa_{\epsilon}\left( t, t' \right) =\kappa \left( t, t' \right), \nu_{\epsilon}  = \nu $. 
		\end{proposition} 
		\begin{proof} 
			The result is straightforward, since the shift does not change the shape of the distribution. 
		\end{proof}

		\begin{proposition} \label{proposition:StP_plus_StP}
			For the case described in Lemma \ref{lemma_StP_plus_StP}, the approximate StP that has the same moments as the SP $\mathcal{SP}_e$  is given as follows 
			\begin{equation}
				\hat{e}(k) \sim  \mathcal{TP}(m_e(k),\kappa_{e}(k, k'), \nu_e ) ,
			\end{equation}
			where 
			\begin{align} 
				m_{e}(k) & = m_{g}(k)+m_{v}(k) , \label{eq:StP+StP:m_{e}(k)}  \\ 
				\kappa_{e}(k, k') & = \frac{\nu_e-2}{\nu_e} \Big( \frac{\nu_g\kappa_{g} \left( k, k' \right) }{\nu_g-2} + \frac{\nu_v \kappa_{v}(k, k')}{\nu_v-2}\Big) , \label{eq:StP+StP:k_{e}(k)} 
			\end{align}
			and the DoF is calculated by the mean of the individual DoFs in our approach
			\begin{equation}
				\nu_e = \frac{\nu_g + \nu_v}{2} .
			\end{equation}
		\end{proposition}
		\begin{proof}
			Eqs. \eqref{eq:StP+StP:m_{e}(k)} and \eqref{eq:StP+StP:k_{e}(k)} are the same as those given in \eqref{eq:e-mean_mg+mv} and \eqref{eq:e-cov_kg+kv}, respectively. The calculation of the covariance functions is given in Appendix \ref{appendx:proof_lemma_StP_plus_StP}. 
		\end{proof}
		
		Given that the measurement noise covariance function $\kappa_v \left( t, t' \right)$ and the DoF are available beforehand, 
		and $\kappa_e \left( t, t' \right)$ can be learned from the data as shown next, the covariance kernel of the TSP $\mathcal{SP}_{x}$ can be calculated as in \eqref{eq:e_kernel}, according to the results of Proposition \ref{proposition:StP_plus_StP}.  
		We use the RBF kernel (\ref{eq:RBF kernel}) again, the same as in section \ref{RGPlearning} to learn the hyperparameters $\boldsymbol{\eta }_k \triangleq \{\varepsilon_{m}, \zeta\}$ simultaneously with estimating the values of function $e(\cdot)$. Following Propositions 1 and 2, similar recursive update framework as that of the RGP* is used, which basically substitutes StP as a GP replacement algorithm. Detailed derivation can be found in \cite{shah2014student,tracey2018upgrading} and is omitted here. Then, the updated T-FoT is the same calculated as in \eqref{eq:estimated_f_mean_GP} and its covariance at time $k$ is given as 
		\begin{equation}
			\mathbb{C}\left[ \hat{f}_k\left( t \right) \right] = \frac{\nu +\lVert \mathbf{y}_k \rVert _{\mathbf{J}_{\mathbf{t}_k}}-2}{\nu +\left| \mathbf{t}_k \right|-2}\left( \kappa\left( k,k \right) -\lVert \mathbf{k}_k \rVert _{\mathbf{J}_{\mathbf{t}_k}} \right) ,
			\label{eq:estimated_f_cov_StP}
		\end{equation}
		where 
		\eqref{eq:StPR_covariance} was used in \eqref{eq:estimated_f_cov_StP}, $\mathbf{k}_k$ is calculated by \eqref{eq:StP+StP:k_{e}(k)}.

		\section{Simulations}
		\label{sec:Simulations}
		The proposed approaches are evaluated and compared with two highly relevant data-driven approaches in four demanding maneuvering target scenarios under two different types of colored measurement noise. Note that since no model information is available beforehand about the target dynamics and even the trajectory may not be produced by an SSM, no traditional filters were simulated for comparison.  
		
		The first comparison method is the so-called GPMT \cite{aftab2020learning}. It assumes that the target trajectory follows a zero-mean sparse GP of which the covariance kernel is learned from the measurement. That is, unlike our DSD approach, it does neither decompose the deterministic trend part of the trajectory nor perform any T-FoT fitting. 
		For some reason or another, this tracker exhibits high sensitivity to data size $d$ \cite{aftab2020learning}, requiring empirical optimization through multiple experiments in various scenarios. In contrast, our approaches work well by simply setting $d=5$ in all scenarios. We conjecture that this is because our GP/StP learning is based on the DSD and only regarding the minor fitting residual part which is much less affected by the target maneuvering of which the major part, that is, the trend of the trajectory, has been handled by the deterministic T-FoT fitting. 
		The second comparison method, called T-FoT, performs the second-order polynomial fitting only as detailed in \cite{li2018joint}, which produces the estimate of T-FoT without any uncertainty. This amounts to the first part of our approaches here which additionally model and learn the residual fitting error by a GP or by an StP, namely T-FoT-GP and T-FoT-StP, respectively. 

		The position root mean square errors (RMSE) as well as the average RMSE (ARMSE) is used as the metrics, which are defined as follows:
		\begin{equation}
			\mathrm{RMSE}_\mathbf{x}^k = \sqrt{\frac{1}{N} \sum_{i=1}^{N} \left( \mathbf{x}_{i,k}-\hat{\mathbf{x}}_{i,k} \right)^2 },
		\end{equation}
		\begin{equation}
			\mathrm{ARMSE} = \frac{1}{T} \sum_{k=1}^{T} \mathrm{RMSE}_\mathbf{x}^k, 
		\end{equation}
		where $N=100$ is the number of simulation trials, 
		$\mathbf{x}_{i,k}$ and $\hat{\mathbf{x}}_{i,k}$ are the true position and its estimate at time $k$ of $i$-th trial, respectively, and $T$ is the total number of measurement/sampling steps. 
		
		
		\subsection{Simulation Scenarios}
		The simulation is carried out in the following four challenging single-sensor single-target scenarios, as shown in Fig. \ref{fig:scenarios}, where the former two use a simple target motion pattern, while the latter two use highly maneuvering trajectories. 
		
		\begin{enumerate}
			\item[S1-]Gradual coordinated turn (CT): the target trajectory combines the CT (15°/s for 10s) with the constant velocity (CV) motion where the initial velocity $v_0$ is randomly initialized within the range of $150\leq v_0\leq 250\ m/s$.  
			\item[S2-]Sharp CT: similar CT motion to scenario S1 but with increased turn rates ($30{^\circ}/{s}$ for $9s$), showing a higher agile dynamics. 
			\item[S3-]Gradual WPV-CT turn: this scenario integrates a Wiener process velocity 
			and a CT model using no position and velocity noise but zero-mean Gaussian turn rate noise with covariance 0.15.  
			\item[S4-]GP: the target motion is modeled using two GPs for the $x$- and $y$- coordinates. Each GP employs a mean function given by a second-order polynomial 
			and an RBF kernel covariance with constant hyperparameters variance $\sigma _{gp}^{2}=1\ m^2 $ and length scale $\zeta _{gp}=1\ s$.
		\end{enumerate}
		
		
		\begin{figure}[!htbp]
			\centering
			\includegraphics[width=0.44\textwidth]{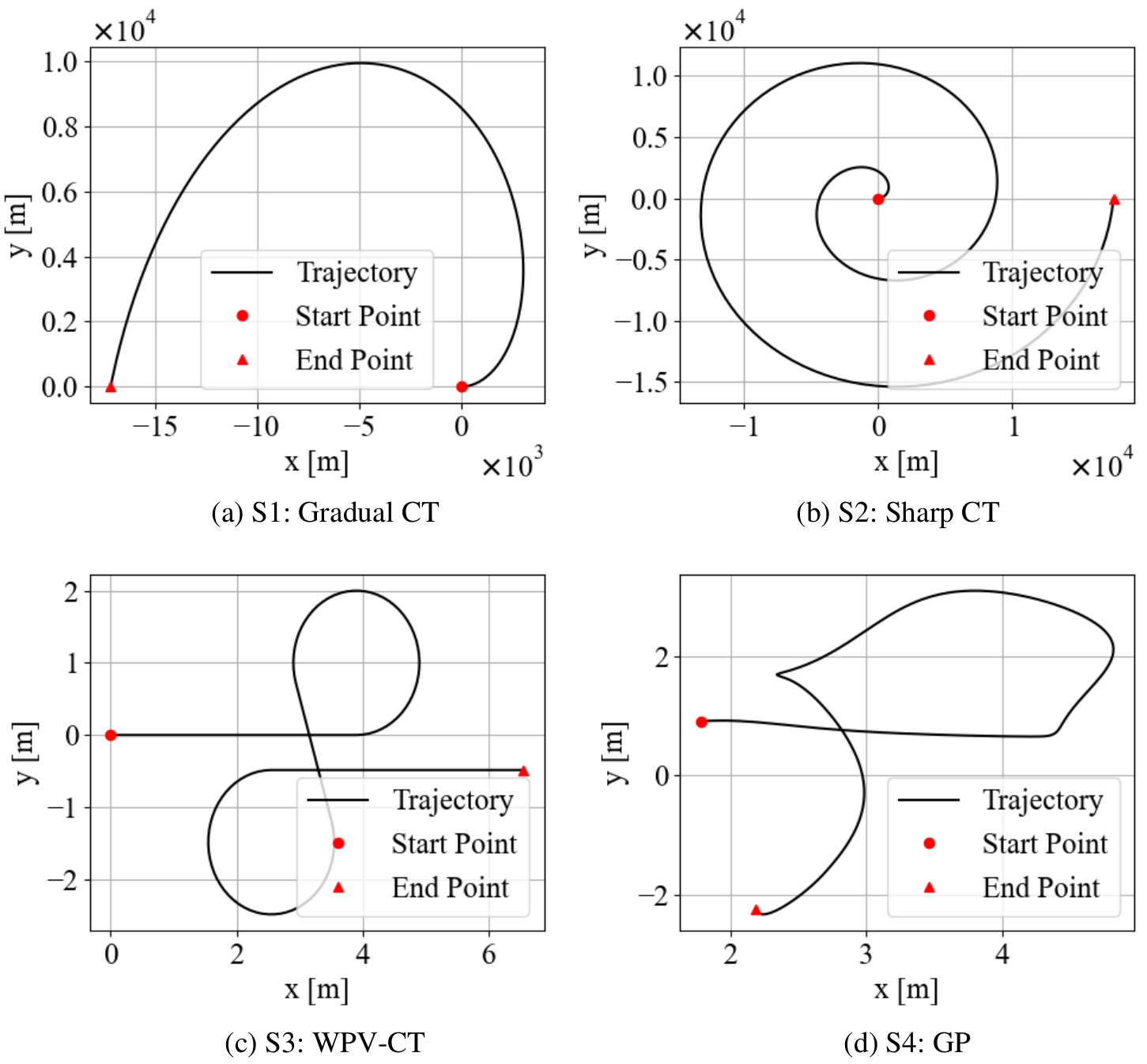}
			\caption{Sample trajectories in four scenarios.}
			\label{fig:scenarios}
		\end{figure}
		

		\begin{figure}[!htbp]
			\centering
			\includegraphics[width=0.92\linewidth]{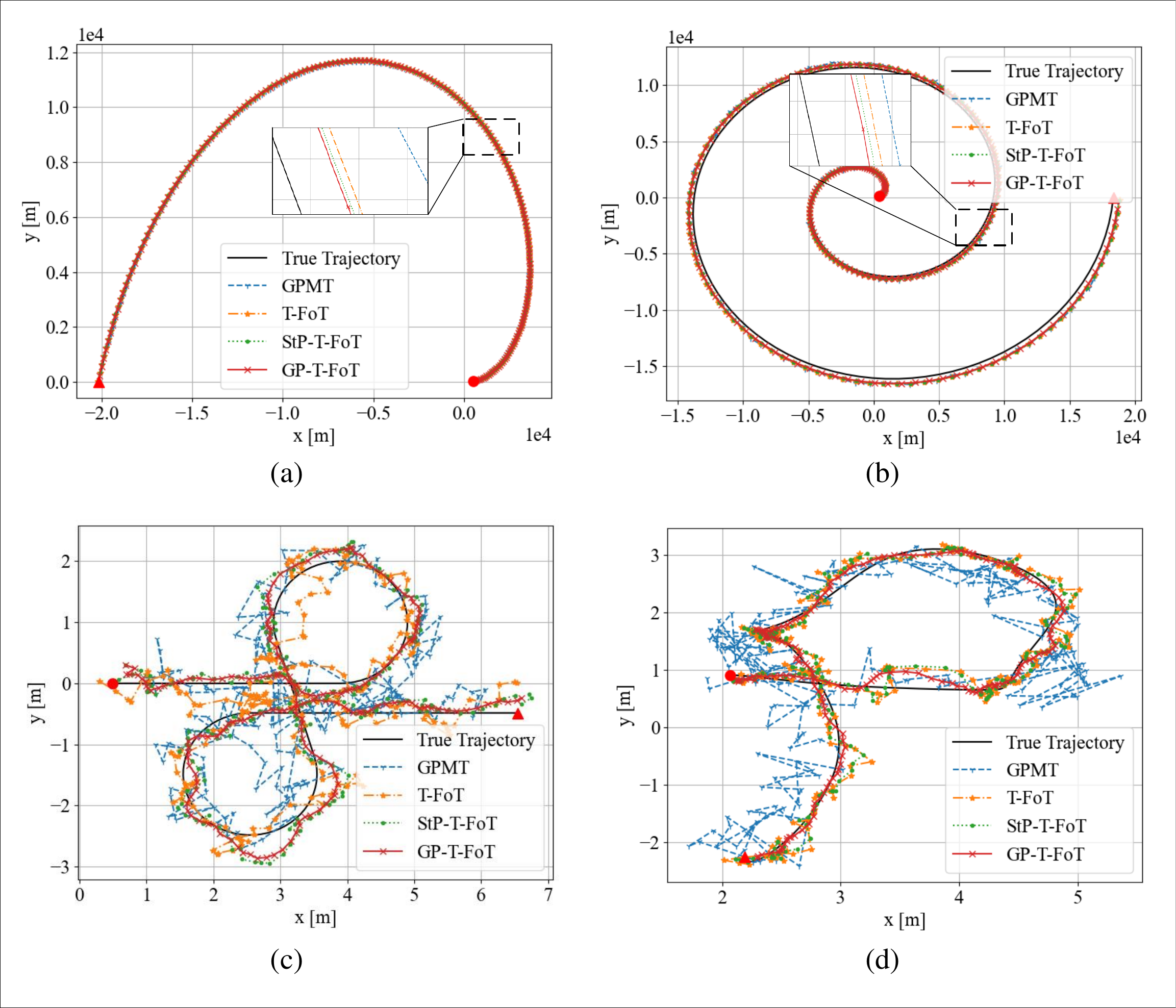}
			\caption{Trajectory estimates in one trial using GP colored measurement noise.}
			\label{fig:Result}
		\end{figure}
		
		Two types of measurement noise that are directly applied on the position measurements are considered as follows. 
		\begin{itemize}
			\item \textbf{GP colored noise.} The measurement noise comes from $\mathcal{G}\mathcal{P}\left( \boldsymbol{0},\kappa \left( t,t' \right) \right)$ with an RBF covariance kernel. The variance of the noise kernel in scenarios S1 and S2 is $\sigma_m^2=25\ m^2$, and the length scale is $\zeta=1\ s $. For scale scenarios S3 and S4 the noise parameters are adjusted to $\sigma_m^2=0.25\ m^2$ and $\zeta=0.1\ s$.  
			\item \textbf{Heavy-tailed colored noise.} The measurement noise is generated through two GPs with the RBF kernel. The input time points are randomly selected from a continuous-time grid with varying intervals between the points. The length scale parameter of the noise kernel is set to $\zeta = 1\ {s}$ scenarios S1 and S2 and $\zeta = 0.1\ {s}$  in scenarios S3 and S4. The variance of the noise kernel $\sigma_m^2$ is defined as follows in order to exhibit heavy-tailed properties \cite{Li23AAt}, in scenarios S1 and S2,  
			\begin{equation}
				\sigma_m^2=
				\left\{ \begin{array}{l}
					25\ m^2,\ \text{with  probability\ } 95\% \\
					250\ m^2,\ \text{with probability\ }5\% \\
				\end{array} \right. 
			\end{equation}
			and in scenarios S3 and S4,  
			\begin{equation}
				\sigma_m^2=
				\left\{ \begin{array}{l}
					0.25\ m^2,\ \text{with  probability\ } 95\% \\
					2.5\ m^2,\ \text{with probability\ }5\% \\
				\end{array} \right. 
			\end{equation}
		\end{itemize}

		\begin{figure}[!htbp]
			\centering
			\includegraphics[width=0.92\linewidth]{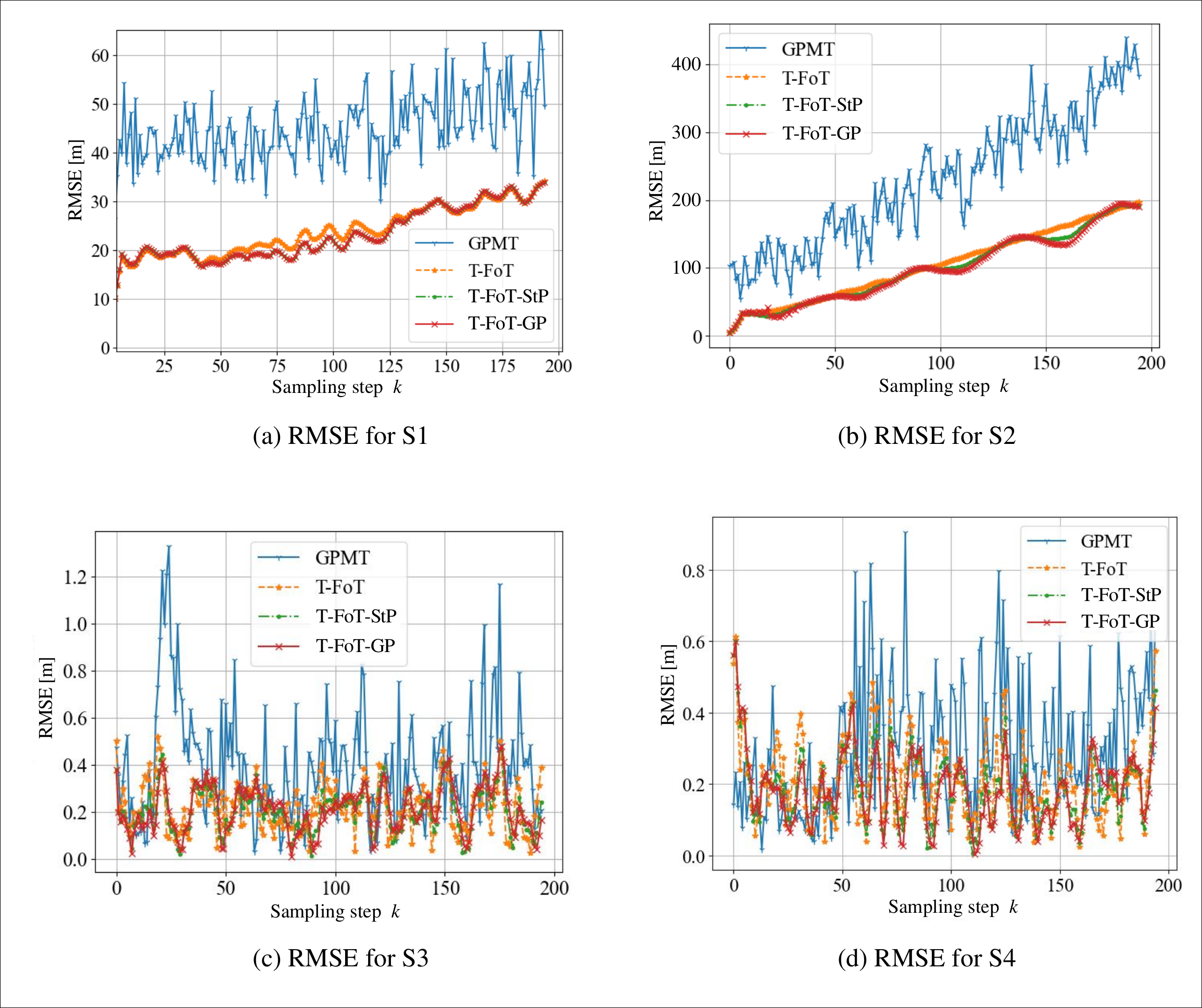}
			\caption{RMSE of 100 trials in the presence of GP colored data noise.}
			\label{fig:RMSE}
		\end{figure}
		

		\subsection{Results}
		\label{sec:simulation results}
		\subsubsection{GP Colored Noise}
		The first series of simulations are conducted in the presence of GP measurement noise. The trajectory estimate, RMSE, ARMSE are given in Figs. \ref{fig:Result}, \ref{fig:RMSE}, and Table \ref{tab:Aver.RMSE_comparison_GP_noise}, respectively.
		As shown, the proposed T-FoT-GP and T-FoT-StP which achieve comparable performance outperform both comparison methods in all scenarios, validating the efficacy of our proposed SP-for-tracking framework. In particular, the superiority of all three T-FoT fitting-based approaches is pronounced in comparison with the GPMT in scenarios S1 and S2 using simple smooth trajectory suitable for polynomial fitting. This confirms the fact that it is overly restrictive and unrealistic to assume stationary and zero mean TSP. DSD is more preferable in these cases when the trajectory has a significant deterministic trend. 
		Moreover, the GPMT displayed significant sensitivity to data size $d$, which was solved through trial-and-error \cite{aftab2020learning}. We tested and used the best one for different scenarios, that is, $d_{S1}=8$, $d_{S2}=8$, $d_{S3}=9$, and $d_{S4}=10$. 
		Although some methods 
		have been given for automatic parameter selection \cite{snelson2005sparse,liu2020gaussian}, $d$-sensitivity remains an open problem. In summary, both T-FoT-GP and T-FoT-StP enhance the algorithm stability and accuracy by decomposing the deterministic and stochastic components, mitigate the potential underfitting from data variability by using an appreciate sliding time window, and compensate the residual fitting error by taking into account the temporal correlation of the data. 
		
		\begin{table}[!h]
			\caption{ARMSE [m] of different trackers in the presence of GP colored measurement noise} 
			\centering
			\begin{tabular}{ c | c c c c }
				\toprule
				\textbf{Tracker}               & S1      & S2     & S3         & S4 \\
				\noalign{\smallskip}\hline\noalign{\smallskip}
				\textbf{GPMT}               & 45.212      & 322.402      & 0.4564         & 0.3837 \\ \noalign{\smallskip}\hline\noalign{\smallskip}
				\textbf{T-FoT}              & 26.846       & 199.678      & 0.2871         & 0.1951 \\ \noalign{\smallskip}\hline\noalign{\smallskip}
				\textbf{T-FoT-GP}           & 24.113       & 177.151      & 0.1964        & 0.1565 \\
				\noalign{\smallskip}\hline\noalign{\smallskip}
				\textbf{T-FoT-StP}          & 24.849       & 187.715      & 0.2096        & 0.1682 \\ \noalign{\smallskip}
				\bottomrule
			\end{tabular}
			\label{tab:Aver.RMSE_comparison_GP_noise}
		\end{table}
		
		\subsubsection{Heavy-tailed Colored Noise}
		The second series of simulations are conducted 
		in the presence of aforementioned heavy-tailed colored noise. 
		Note that the same parameter $d$ as in the first set of simulations is used in the GPMT, which is not guaranteed to be the best. For T-FoT-StP, the DoF parameter is set to $\nu = 5$ which is kept constant for simplicity. 
		The RMSE and ARMSE are given in Fig. \ref{fig:RMSE_heavy-tailed} and Table \ref{tab:Aver.RMSE_comparison_heavy-tailed}, respectively. As shown, the proposed T-FoT-GP and T-FoT-StP show significant advantages again in all scenarios. 
		In particular, 
		our proposed methods maintain stable tracking accuracy during the steady phases of the target trajectory and even when the target maneuvers are taken. 
		In contrast, the performance of GPMT and T-FoT deteriorates more or less in the latter case.

		\begin{figure}[!htbp]
			\centering
			\includegraphics[width=0.92\linewidth]{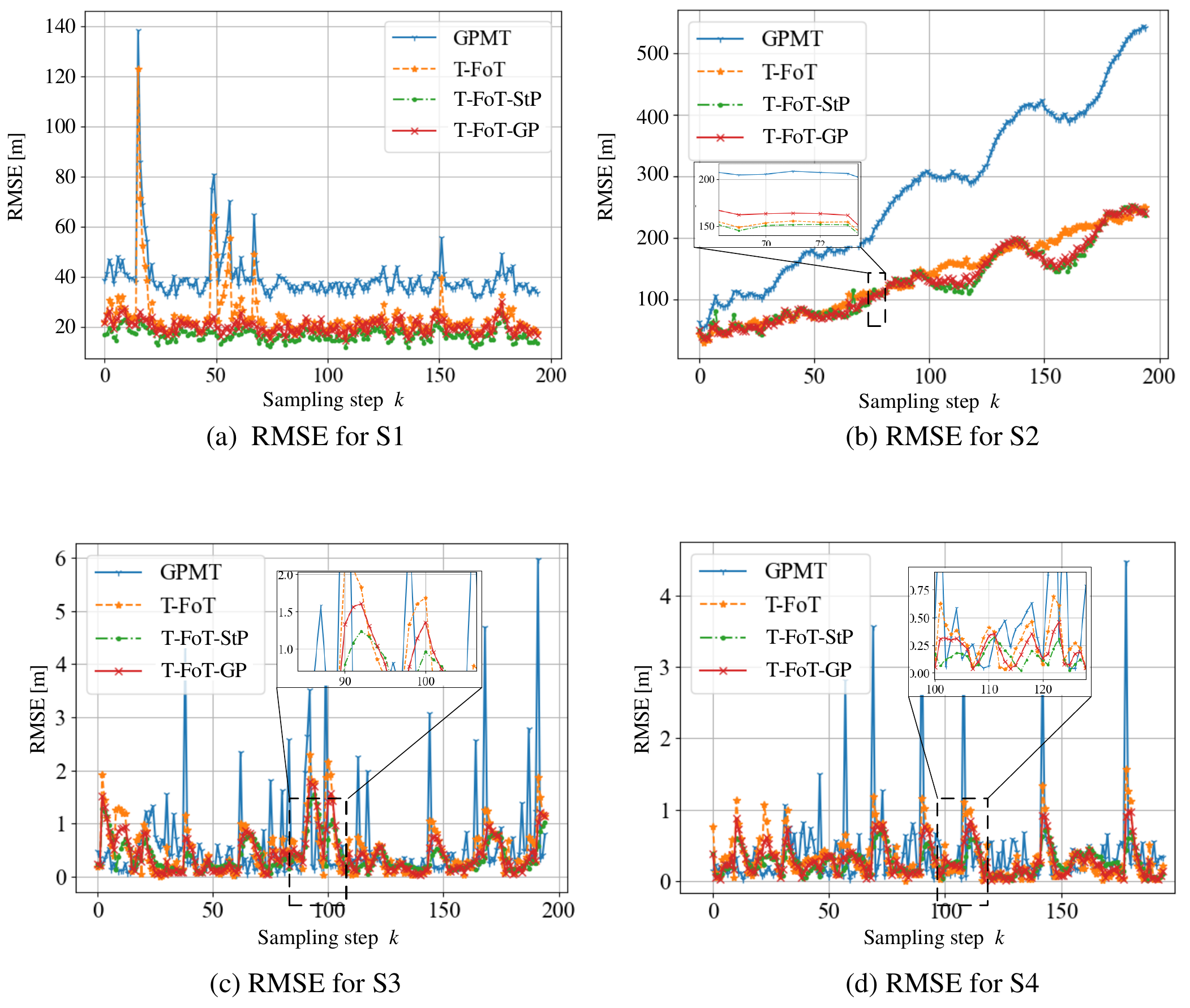}
			\caption{RMSE of 100 trials in the presence of heavy-tailed colored measurement noise.}
			\label{fig:RMSE_heavy-tailed}
		\end{figure}
		
		Moreover, 
		the T-FoT-StP demonstrates even more superior performance in each scenario 
		in comparison with T-FoT-GP and the others. 
		This can be attributed to two fundamental characteristics inherent in the StP. First, its heavy-tailed distribution enables more effective accommodation of these outliers caused by occasionally significant noise. 
		Secondly, its self-similarity and long-memory properties allow for superior capture of temporal patterns in complex dynamical environments.

		\begin{table}[!h]
			\caption{ARMSE [m] of different trackers in the presence of heavy-tailed colored measurement noise} 
			\centering
			\begin{tabular}{ c | c c c c }
				\toprule
				\textbf{Tracker}               & S1      & S2     & S3         & S4 \\
				\noalign{\smallskip}\hline\noalign{\smallskip}
				\textbf{GPMT}               & 49.181      & 359.108      & 0.768         & 0.657 \\ \noalign{\smallskip}\hline\noalign{\smallskip}
				\textbf{T-FoT}              & 28.711       & 208.398      & 0.495        & 0.586 \\ \noalign{\smallskip}\hline\noalign{\smallskip}
				\textbf{T-FoT-GP}           & 25.953       & 189.395      & 0.358        & 0.473 \\
				\noalign{\smallskip}\hline\noalign{\smallskip}
				\textbf{T-FoT-StP}          & 21.469       & 176.893      & 0.287        & 0.295 \\ \noalign{\smallskip}
				\bottomrule
			\end{tabular}
			\label{tab:Aver.RMSE_comparison_heavy-tailed}
		\end{table}
		

		\section{Conclusion}
		\label{sec:Conclusion}
		We propose to model the target states in time series, as well as the measurement noise, as an SP and reformulate the tracking problem as an online trajectory SP fitting and learning paradigm. A flexible DSD-based SP modeling and learning framework has been presented, which consists of two sequential stages. In the first stage, a deterministic function is employed to fit the deterministic trend of the trajectory. In the second stage, a novel SP is used to estimate the residual fitting error. This SP can be approximated by either a GP or an StP. The proposed method yields a Markov-free data-driven tracking approach capable of generating continuous-time trajectories with minimal prior knowledge of the target dynamics. 
		Additionally, the proposed method models and learns not only the temporal correlation of states in the time series but also that of the measurement noises using GP or StP. 
		This further enhances the robustness and compatibility of the approach in complex noisy environments. The effectiveness and advantages of the proposed approach have been demonstrated in diverse simulation environments: from simple target motion patterns to highly maneuverable trajectory, spanning both large- and small-scale scenarios while accommodating outlier noise. Our future work will consider the extension of the DSD-based SP-for-tracking framework to the multisensor system \cite{Jingyuan2025GP-AA} and to the nonlinear and even non-fully observable measurement model. 
		
		

		\appendix
		\label{sec:appendix}
		\subsection{Proof of Lemma \ref{lemma_linear_Function_GP} } \label{appendx:prof_lemma_linear_Function_GP}
		The mean function of the pseudo measurement SP $g(k)$ is calculated as 
		\begin{align}
			\mathbb{E}[g(k)]
			&=\mathbb{E}[{\mathbf{H}} \epsilon(k) ] \nonumber  \\
			&={\mathbf{H}} \mathbb{E}[\epsilon(k) ] \nonumber  \\
			&= {\mathbf{H}} m(k) \label{eq:GP_linearTrasf_mean}
		\end{align}
		The covariance function of $g(k)$ is 
		\begin{align}
			\kappa_g(k, k')
			&=\mathbb{E}[(g(k)-\mathbb{E}[g(k)])(g(k')-\mathbb{E}[g(k')])^{\top}] \nonumber  \\
			&={\mathbf{H}}\mathbb{E}[(\epsilon(k)-m(k))(\epsilon(k')-m(k'))] {\mathbf{H}}^{\top}\nonumber  \\
			&={\mathbf{H}}\kappa(k, k'){\mathbf{H}}^{\top} \label{eq:GP_linearTrasf_cova}
		\end{align}
		
		\subsection{Proof of Lemma \ref{lemma_GP-mt}} \label{appendx:prof_lemma_GP-mt} 
		The mean function of $\epsilon(k)$ is calculated as 
		\begin{align}
			\mathbb{E}[\epsilon(k)]
			&=\mathbb{E}[f(k)-m(k)] \nonumber \\
			&=F(t;\mathbf{C})-m(k) \nonumber  \\
			&= \mathbf{0} \label{eq:mean_f-m=0}
		\end{align}
		The covariance function of $\epsilon(k)$ is 
		\begin{align}
			\mathbb{C}_\epsilon(k, k') 
			&=\mathbb{E}[(\epsilon(k')-\mathbb{E}[\epsilon(k')])(\epsilon(k)-\mathbb{E}[\epsilon(k)])] \nonumber  \\
			&=\mathbb{E}[(f(k')-m(k'))(f(k)-m(k))]\nonumber  \\
			&=\kappa(k, k')
		\end{align}
		
		\subsection{Proof of Lemma \ref{lemma_GP_plus_GP}}
		\label{appendx:proof_lemma_GP_plus_GP}
		The mean function of $e(t)$ is calculated as 
		\begin{align}
			\mathbb{E}\left[ e \left( t \right)\right]
			&=\mathbb{E}\left[ g \left( t \right) +v\left( t \right) \right]  \nonumber \\
			&=m_{g}\left( t \right) +m_{v}\left( t \right) \label{eq:e-mean_mg+mv}
		\end{align}
		The covariance function of $e(t)$ is 
		\begin{align}
			\kappa_e(t,t')
			=&\mathbb{C}_e(t,t')    \nonumber \\
			=&\mathbb{E} \big[ \left( e\left( t \right)-\mathbb{E}\left[ e\left( t \right) \right] \right) \left( e\left( t' \right)-\mathbb{E}\left[ e\left( t' \right) \right] \right) \big]  \nonumber \\
			=&\mathbb{E}\big[ \left( g \left( t \right) + v\left( t \right)-\left( m_{g}\left( t \right) + m_{v}\left( t \right) \right) \right)  \nonumber \\
			& \left( g \left( t' \right) + v\left( t' \right)-\left( m_{g}\left( t' \right) + m_{v}\left( t' \right) \right) \right) \big] \nonumber \\
			=&\mathbb{E}\big[ \left( g \left( t \right)-m_{g}\left( t \right) + v\left( t \right)-m_{v}\left( t \right) \right) \nonumber \\
			& \left( g \left( t' \right)-m_{g}\left( t' \right) + v\left( t' \right)-m_{v}\left( t' \right) \right) \big]  \nonumber \\
			=&\mathbb{E}\left[ \left( g \left( t \right)-m_{g}\left( t \right) \right) \left( g \left( t' \right)-m_{g}\left( t' \right) \right) \right] \nonumber \\
			& + \mathbb{E} \big[\left( v\left( t \right)-m_{v}\left( t \right) \right) \left( v\left( t' \right)-m_{v}\left( t' \right) \right) \big]  \label{eq:e-cov_kg+kv_independence} \\
			=&\kappa_{g}\left( t,t' \right) + \kappa_{v}\left( t,t' \right) \label{eq:e-cov_kg+kv}
		\end{align}
		where independence between $g(\cdot)$ and $v(\cdot)$ was used in \eqref{eq:e-cov_kg+kv_independence}. 
		
		\subsection{Proof of Lemma \ref{lemma_StP_plus_StP}}
		\label{appendx:proof_lemma_StP_plus_StP}
		The covariance function of $e(t)$ is 
		\begin{align}
			\kappa _e\left( t,t' \right) &=\frac{\nu _e\left( t \right) -2}{\nu _e\left( t \right)}\mathbb{C}\left( e\left( t \right) ,e\left( t' \right) \right) \nonumber \\
			&=\frac{\nu _e\left( t \right) -2}{\nu _e\left( t \right)}\left( \mathbb{C}\left( g\left( t \right) ,g\left( t' \right) \right) +\mathbb{C}\left( v\left( t \right) ,v\left( t' \right) \right) \right)  \label{eq:e-cov_kg+kv_StP-independence}\\
			&=\frac{\nu _e\left( t \right) -2}{\nu _e\left( t \right)}\left( \frac{\nu _g\left( t \right)\kappa _g\left( t,t' \right)}{\nu _g\left( t \right) -2} +\frac{\nu _v\left( t \right)\kappa _v\left( t,t' \right)}{\nu _v\left( t \right) -2} \right) . 
			\nonumber 
		\end{align}
		where independence between $g(\cdot)$ and $v(\cdot)$ was used in \eqref{eq:e-cov_kg+kv_StP-independence}.

\end{document}